\documentclass[final,1p,times,number]{elsarticle}
\usepackage[colorlinks]{hyperref}

\usepackage{amssymb}

\usepackage{amsfonts,amsmath,amscd,amsthm,color}
\usepackage[all]{xy}

\newfont{\eufm}{eufm10 scaled\magstep1}

\newcommand{\cB}{\mathcal{B}}
\newcommand{\cC}{\mathcal{C}}
\newcommand{\cH}{\mathcal{H}}
\newcommand{\cL}{\mathcal{L}}
\newcommand{\cF}{\mathcal{F}}

\newcommand{\cR}{\mathcal{R}}
\newcommand{\cS}{\mathcal{S}}
\newcommand{\cV}{\mathcal{V}}
\newcommand{\cW}{\mathcal{W}}

\newcommand{\bbN}{\mathbb{N}}
\newcommand{\bbZ}{\mathbb{Z}}
\newcommand{\bbC}{\mathbb{C}}

\def\dres{\partial{\rm Res}}

\def\ord{\rm ord}

\def\KdV{\mathbf{KdV}}

\def\gcrd{\rm gcrd}

\def\kdv{\rm kdv}

\def\para{\vspace{1.5 mm}}

\newtheorem{thm}{Theorem}[section]
\newtheorem{lem}[thm]{Lemma}
\newtheorem{cor}[thm]{Corollary}
\newtheorem{prop}[thm]{Proposition}
\newtheorem{defi}[thm]{Definition}
\newtheorem{rem}[thm]{Remark}

\newtheorem{ex}[thm]{Example}
\newtheorem{alg}[thm]{Algorithm}

\begin{document}

\begin{frontmatter}

\title{Factorization of KdV Schr\" odinger operators\\ using differential subresultants}

\author{Juan J. Morales-Ruiz}
\address{Dpto. de Matem\' atica Aplicada. E.T.S. Edificaci\' on. Avda. Juan de Herrera 6.\\
Universidad Polit\' ecnica de Madrid. 28040, Madrid. Spain}
\ead{juan.morales-ruiz@upm.es}

\author{Sonia L. Rueda}
\address{Dpto. de Matem\' atica Aplicada. E.T.S. Arquitectura. Avda. Juan de Herrera 4.\\
Universidad Polit\' ecnica de Madrid. 28040, Madrid. Spain}
\ead{sonialuisa.rueda@upm.es}

\author{Maria-Angeles Zurro}
\address{Dpto. de Matem\' aticas. Facultad de Ciencias. Ciudad Universitaria de Cantoblanco.\\
Universidad Aut\'onoma de Madrid. E-28049 Madrid. Spain}
\ead{mangeles.zurro@uam.es}



\begin{abstract}
We address the classical factorization problem of a one dimensional Schr\"odinger operator $-\partial^2+u-\lambda$, for a stationary potential $u$ of the KdV hierarchy but, in this occasion, a "parameter" $\lambda$. Inspired by the more effective approach of Gesztesy and Holden to the "direct" spectral problem, we give a symbolic algorithm by means of differential elimination tools to achieve the aimed factorization. Differential resultants are used for computing spectral curves, and differential subresultants to obtain the first order common factor. To make our  method fully effective, we design a symbolic  algorithm to compute the integration constants of the KdV hierarchy, in the case of KdV potentials that become  rational under a Hamiltonian change of variable. Explicit computations are carried for Schr\"odinger operators with solitonic potentials.
\end{abstract}

\begin{keyword}
Schr\" odinger operator \sep factoriazation of ODOs \sep  differential resultant \sep differential subresultant \sep spectral curve 
\MSC[2010]  13P15, 12H05
\end{keyword}

\end{frontmatter}



\section{Introduction}\label{sec-Introduction}

This paper addresses the effective factorization of the Schrodinger operator
\begin{equation}
L- \lambda = - \partial^2 + u - \lambda
\end{equation}
for a stationary potential $ u $ in a complex variable, say $ x $, and $ \lambda $ a parameter over the field of coefficients. It is well known that whenever the potential satisfies one of the differential equations of the Korteweg de Vries (KdV) hierarchy, this problem is intimately related to the existence of a plane  algebraic curve $\Gamma$, the spectral curve associated to the operator $L$. In 1923, J.L. Burchnall and  T.W. Chaundy \cite{BC} established a corres\-pondence between commuting differential operators and algebraic curves.
They discovered the {\sf spectral curve}, defined by the so called Burchnall and Chaundy (BC) polynomial. This discovery allowed an algebro-geometric approach to handling the direct and inverse spectral problems for the {\sf finite-gap} operators, with the spectral data being encoded in the spectral curve and an associated line bundle \cite{K77}. In particular, KdV Schr\" odinger operators (a special case of finite-gap operators) can be treated by the methods in \cite{K77}, but in this paper we present a different approach to the direct spectral problem inspired by the more effective treatment of Gesztesy and Holden in \cite{GH}. We advice for instance \cite{GH} for a historic introduction on the subject.

\para

Classically, the spectral curve was computed using a Lenard-type differential recursion (see \cite{GH}), where arbitrary integration constants appeared at each step of the iterative process. In \cite{GH} Theorem D.1, the intimate relationship between these integration constants and $\Gamma$ is shown.  Our approach to the problem of computation of constants has the goal of designing an algorithm that depends only on the potential $u$, but not directly on the spectral curve.
For this purpose we describe the flag structure that the constants create, see Section \ref{subsec-flag}.
In the case of potentials that become rational under a Hamiltonian change of variable \cite{AMW}, we have been able to design the aimed algorithm.

\para

Based on Goodearl's theoretical results \cite{Good}, we describe the centralizer of $L$  for a KdV potential $u$. In other words, we determine the essential odd order operator $A_{2s+1}$ of the centralizer of $L$, that together with $L$ generates the centralizer as a $\bbC$-algebra $\bbC[L,A_{2s+1}]$. The potential $u$  satisfies a fixed $\KdV_s$ equation
$$
\KdV_s=\kdv_s+c_1\kdv_{s-1}+\cdots +c_s\kdv_0 =0,
$$
with corresponding integration constants $c_1,\ldots ,c_s$ in $\mathbb{C}$. Thus for a fixed potential $u$, the algorithmic determination of the operator  $A_{2s+1}$ relies on the algorithmic computation of the constants $c_i$.

Once we have explicitly obtained the operator $ A_ {2s + 1} $, a defining equation  for $\Gamma$ can be computed. In fact, E. Previato \cite{Prev} used differential resultants to compute spectral curves,  opening the door to symbolic computation techniques. The use of these techniques did not transcend so far \cite{GH}, \cite{Mir} and their defining  polynomials are commonly computed as characteristic polynomials \cite{GH}, \cite{Brez3}, \cite{Mir}.  Differential resultants for ordinary differential operators were defined in the 90's by Berkovich and Tsirulik \cite{BT} and studied by Chardin \cite{Ch}, who also defined the differential subresultant sequence;  see \cite{McW} for a recent report on the subject.

\para

We present a new symbolic algorithm for the factorization of a KdV Schr\" odinger operator $L-\lambda$ over the field $K(\Gamma)$  of its spectral curve $\Gamma$ using differential subresultants. There are other factorization algorithms for linear ordinary differential operators in the literature, as \cite{BrP}, \cite{VanHoeij}, \cite{Singer}. But we benefit from the use of the first subresultant since it provides a differential  algebraic formula only in terms of the potential $ u $ and the computed constants. In this way the factorization obtained for $ L- \lambda $ is written as
\[
L-\lambda
=(-\partial-{\phi})(\partial-{\phi}) \leqno{(\star)}
\]
with ${\phi}$  a quotient of two determinants of matrices with entries differential polynomials in $u$. Whenever the spectral curve admits a global parametrization, the algebraic framework that justifies the correctness of the algorithms allows to develop a parametric version of $(\star)$. In the examples of Section \ref{sec-Ejemplos} we illustrate some of these cases.

\para

A very important requirement in this work is to treat $\lambda$ as a parameter. The differential operator $L-\lambda$ is treated first as an operator with coefficients in the field $K(\lambda )$; then, when the spectral curve $\Gamma$ is considered, as a differential operator with coefficients in the field $K(\Gamma)$ of rational functions on $\Gamma$.
Our symbolic factorization structure allows the specialization process to points $(\lambda_0,\mu_0)$ on $\Gamma$, recovering the classical factorization of $L-\lambda_0$ at each point $(\lambda_0,\mu_0)$ of $\Gamma_s$, see  \cite{GH}.
Another approach to the factorization is carried by means of Darboux transformations and the raising and lowering operators $A^+$ and $A^-$, but with this approach one previously needs to compute a set of solutions for a finite set of energy levels, see \cite{AMW} and the references therein.

\para

The paper is organized as follows.
In Section \ref{sec-Preliminaries}, we construct the KdV hierarchy and define differential subresultants, reviewing its main properties. Section \ref{sec-constants} contains our  algorithm for computation of the integration constants of the KdV hierarchy. Then in Section \ref{sec-spectralcurves}  we describe the centralizer of $L$ and  compute the operator $A_{2s+1}$. We are ready to review Previato's Theorem, applying it to the computation of the spectral curve of the Lax pair $\{L,A_{2s+1}\}$. Section \ref{sec-KdV factors} contains our factorization algorithm for $L-\lambda$ as an operator in $K(\Gamma)[\partial]$.  In Section \ref{sec-Ejemplos}, we apply our algorithms to three special families of solitons. A parametric version  of the factors is also included for those examples.

\para

We implemented the algorithm for the computation of constants and the factorization algorithm using Maple 18. We used these implementations to compute the examples in Section \ref{sec-Ejemplos}.

\section{Notation}\label{sec-notation}

\para

We establish some notation to be used throughout the whole manuscript.
\para

Let $\bbN$ be the set of positive integers including $0$.
For concepts in differential algebra we refer the reader to \cite{CH}, \cite{VPS} or \cite{Morales}.
Let $K$ be a differential field of characteristic zero with derivation $\partial$, whose field of constants $C$ is algebraically closed.
Let us consider algebraic variables $\lambda$ and $\mu$ with respect to $\partial$. Thus $\partial \lambda=0$ and $\partial \mu=0$ and
we can extend the derivation $\partial$ of $K$ to the polynomial ring $K[\lambda, \mu]$ by
\begin{equation}\label{eq-derpol}
\partial \left(\sum a_{i,j} \lambda^i\mu^j\right)=\sum \partial (a_{i,j}) \lambda^i\mu^j,\,\,\, a_{i,j}\in K.
\end{equation}
Hence $(K[\lambda, \mu],\partial)$ is a differential ring whose ring of constants is $C[\lambda, \mu]$.

\para

 Given a differential commutative ring $E$ with derivation $\partial$,
let us denote by $E[\partial]$ the ring of differential operators with coefficients in $E$ and commutation rule
\[[\partial,a]=\partial a-a\partial=\partial(a), a\in E,\]
where $\partial a$ denotes the product in the noncommutative ring  $E[\partial]$ and $\partial (a)$ is the image of $a$ by the derivation map. The ring of pseudo-differential operators in $\partial$ will be denoted by $E[\partial^{-1}]$  (see \cite{Good})
\[E[\partial^{-1}]=\left\{\sum_{i=-\infty}^n a_i\partial^i\mid a_i\in E, n\in\bbZ\right\},\]
where $\partial^{-1}$ is the inverse of $\partial$ in $E[\partial^{-1}]$, $\partial^{-1}\partial=\partial\partial^{-1}=1$.

\section{Formal KdV Schr\" odinger operators}\label{sec-Preliminaries}

Let us consider a differential indeterminate $u$ over $C$. We will call {\it formal Schr\" odinger operator} to $L(u)=-\partial^2+u$ with coefficients in the ring of differential polynomials $$C\{u\}=C[u,u',u'',\ldots]$$ where $u'$ stands for $\partial (u)$ and $u^{(n)}=\partial^n (u)$, $n\in\bbN$.

\para

\subsection{KdV polynomials and their Lax pair representations}\label{sec-KdV}

In this section we will work with the formal  Schr\" odinger operator $L=L(u)$.
In a convenient way to be used in this paper, we review well known algorithms to compute the differential polynomials in $u$ of the KdV hierarchy and the family of differential operators of its Lax representation. This was studied for the first time in the paper \cite{GD}. We follow the normalization in \cite{GH}, see also \cite{Olver} for other presentations.

\para

Let us consider the pseudo-differential operator
\begin{equation}\label{eq-recursion}
\cR=-\frac{1}{4}\partial^2+u+\frac{1}{2}u'\partial^{-1}\mbox{ and its adjoint }
\cR^*=-\frac{1}{4}\partial^2+u-\frac{1}{2}\partial^{-1}u'.
\end{equation}
Observe that $\cR^*=\partial^{-1}\cR\partial$.
The operator $\cR^*$ is a recursion operator of the KdV equation (see \cite{Olver}, p. 319).
Applying the recursion operator $\cR$, we define:
\begin{equation}\label{eq-kdv}
\kdv_0:=u',\,\,\, \kdv_n:=\cR(\kdv_{n-1}),\mbox{ for }n\geq 1.
\end{equation}
Applying $\cR^*$ we define:
\begin{equation}\label{eq-fn}
v_0:=1,\,\,\, v_n:=\cR^*(v_{n-1}),\mbox{ for }n\geq 1.
\end{equation}
Hence for $n\in\bbN$ it holds
\begin{equation}\label{eq-vkdv}
2\partial(v_{n+1})=\kdv_n.
\end{equation}

We will prove next that for all $n$, $\kdv_n$ and $v_n$ are differential polynomials in $u$, elements of $C\{u\}$. The proof is similar to the one of \cite{Olver}, Theorem 5.31 but we include details for completion,  and due to its importance for their symbolic computation, see Section \ref{sec-formal ex}.
We will call the differential polynomials $\kdv_n$ the {\it KdV differential polynomials}.

\begin{lem}\label{lem-kdv}
The formulas for  $\kdv_n$ and $v_n$ give differential polynomials in $C\{u\}$.
\end{lem}
\begin{proof}
Observe that $\cR(\kdv_{n-1})$ is well defined if and only if $\kdv_{n-1}$ is a total derivative. We will prove this by induction on $n$.
It is trivial for $n=1$ since $\kdv_0=\partial (u)$. Let us assume that $\kdv_{n-1}=\partial (g_{n-1})$, $g_{n-1}\in C\{u\}$.

Since $\cR$ and $\cR^*$ are adjoint operators we have $p\cR (q)=q\cR^*(p)+\partial (a)$, $p,q,a\in C\{u\}$. Thus for $p=u$ and $q=u'$ we get
\[u\cR^{k}(u')=u'(\cR^*)^k(u)+\partial (a_{k}),\,\,\,\mbox{ for }a_{k}\in C\{u\}.\]
Then
\[u \kdv_{n-1}=(\partial u-u\partial) (\cR^*)^{n-1}(u)+\partial (a_{n-1})=-u\partial (\cR^*)^{n-1}(u)+\partial (b),\,\,\, b\in C\{u\}\]
which implies that $u \kdv_{n-1}=\partial (b/2)$ is the total derivative of a differential polynomial in $C\{u\}$. Since
\[\cR=\partial\left(-\frac{1}{4}\partial+\frac{1}{2}\partial^{-1}u+\frac{1}{2}u\partial^{-1}\right)\]
we obtain that $\kdv_n=\cR(\kdv_{n-1})$ is a total derivative.

The fact that $\kdv_n$ is a total derivative and \eqref{eq-vkdv} imply that $v_{n+1}$ are also elements of $C\{u\}$.
\end{proof}

As in \cite{GH}, we define a family of differential operators in $C\{u\}[\partial]$ of odd order (see also \cite{Dikii}, \cite{Novikov})
\begin{equation}\label{eq-A2s+1}
P_1:=\partial,\,\,\,P_{2n+1}:=v_{n}\partial-\frac{1}{2}\partial(v_{n})+P_{2n-1}L,\mbox{ for }n\geq 1.
\end{equation}
Observe that

\begin{equation}\label{eq-P2n+1}
P_{2n+1}=\sum_{l=0}^n\left(v_{n-l}\partial-\frac{1}{2}\partial(v_{n-l})\right)L^l.
\end{equation}
The operators $P_{2n+1}$ have the important property that the commutator  $[P_{2n+1},L]$ is a differential operator in $C\{u\}[\partial]$ but Lemma \ref{lem-fs} shows that it has order zero, it is the multiplication operator by the $\kdv_n$ differential polynomial. This is the famous Lax representation of $\kdv_n$, see \cite{GH}, \cite{Novikov}.
We will call the differential operators $P_{2n+1}(u)$ the {\it KdV differential operators}.

\begin{lem}\label{lem-fs}
For $n\in \bbN$ it holds $[P_{2n+1},L]=\kdv_n$.
\end{lem}
\begin{proof}
One can easily check that $v_1=\cR^*(1)=u/2$ and $[P_1,L]=u'=2\partial (v_1)$. We prove the result by induction on $n$.
Since $P_{2n+3}=O+P_{2n-1}L$ where $O=v_{n+1}\partial-\frac{1}{2}\partial (v_{n+1})$, we have
\[[L,P_{2n+3}]=[L,O]+[L,P_{2n+1}]L=[L,O]-2\partial (v_{n+1})L,\]
with
\begin{align*}
&[L,O]=-2\partial (v_{n+1})\partial^2+(1/2)\partial^3(v_{n+1})-v_{n+1}u',\\
&2\partial (v_{n+1})L=-2\partial (v)'_{n+1}\partial^2+2\partial (v_{n+1})u.
\end{align*}
Observe that $\cR^*=-\frac{1}{4}\partial^{-1}S$ where $S=\partial^3-4u\partial-2u'$. Thus
\[[L,P_{2n+3}]=(1/2)\partial^3(v_{n+1})-v_{n+1}u'-2\partial (v_{n+1})u=(1/2)S(v_{n+1})=
-2\partial\cR^*(v_{n+1})=-2\partial(v_{n+2}).\]
By \eqref{eq-vkdv} the result is proved.
\end{proof}

Now let us consider algebraic indeterminates $c_n$, $n\geq 1$ over $C$. We
define an extended family of {\it KdV differential polynomials} $\KdV_n(u,c^n)$, $n\in\bbN$ in the differential indeterminate $u$ and the list of algebraic indeterminates $c^n=(c_1,\ldots ,c_n)$.
\begin{equation}\label{eq-KdV}
\KdV_0:=u',\,\,\, \KdV_n:=\kdv_n+\sum_{l=0}^{n-1} c_{n-l} \kdv_l, \mbox{ for } n\geq 1
\end{equation}
and an extended family of {\it KdV differential operators} whose coefficients are differential polynomials in $u$ and $c^n$,
\begin{equation}\label{eq-Ac2s+1}
\hat{P}_1:=\partial\mbox{ and }\hat{P}_{2n+1}:=P_{2n+1}+\sum_{l=0}^{n-1} c_{n-l}P_{2l+1}, \mbox{ for } n\geq 1.
\end{equation}
One can easily check that
\begin{equation}\label{eq-PLKdV}
[\hat{P}_{2n+1},L]=\KdV_n=2\partial(f_{n+1}),
\end{equation}
for
\begin{equation}
f_0:=v_0=1\mbox{ and } f_n:=v_n+ \sum_{l=0}^{n-1} c_{n-l}v_l, \mbox{ for } n\geq 1.
\end{equation}

\subsection{Differential resultant and subresultants}\label{sec-Differential Resultant}

\vspace{2mm}

Let $K$ be a differential field as in Section \ref{sec-Preliminaries}.
Let us consider differential operators $P$ and $Q$ in $K[\partial]$ of orders $n$ and $m$ respectively and leading coefficients $a_n$ and $b_m$. We are interested in the common solutions of the system of linear differential equations
\[\left\{\begin{array}{l}P=0\\Q=0\end{array}\right..\]
The tools we have chosen to study this problem are differential resultant and subresultants.
They are an adaptation of the algebraic resultant of two algebraic polynomials in one variable to a noncommutative situation.
We summarize next the definition and some important properties of differential resultants to be used in this work.

\subsubsection{Differential resultant for ODO's and main properties}

\para

The Sylvester matrix $S_0(P,Q)$ is the coefficient matrix of the extended system of differential operators
\[\Xi_0(P,Q)=\{\partial^{m-1}P,\ldots \partial P, P, \partial^{n-1}Q,\ldots ,\partial Q, Q\}.\]
Observe that $S_0(P,Q)$ is a squared matrix of size $n+m$ and entries in $K$. We define the {\it differential resultant} of $P$ and $Q$  to be
$$\dres(P,Q):=\det (S_0(P,Q)).$$

\begin{ex}
Given $P=a_2\partial^2 +a_1\partial +a_0$ and $Q=b_3\partial^3+b_2\partial^2 +b_1\partial +b_0$ in $K [\partial]$,
\[ S_0(P,Q)=
\left[\begin{array}{ccccc}
a_2&a_1+2\partial(a_2)&a_0+2\partial(a_1)+\partial^2(a_2)&2\partial(a_0)+\partial^2(a_1)&\partial^2(a_0)\\
0  &a_2 & a_1 +\partial (a_2)& a_0+\partial (a_1)  &\partial (a_0)\\
0  &   0 & a_2 & a_1    & a_0 \\
b_3&b_2 +\partial (b_3)& b_1 +\partial (b_2)& b_0+\partial (b_1)  &\partial (b_0)\\
0  &   b_3 & b_2 & b_1    & b_0
\end{array}\right].
\]
\end{ex}

The next propositions state the most relevant properties of the differential resultant.

\begin{prop}[\cite{Ch}]\label{prop-Propertiesdres}
Let $(P,Q)$ be the left ideal generated by $P,Q$ in $K[\partial]$.
\begin{enumerate}
\item $\dres(P,Q)=AP+BQ$ with $A,B\in K[\partial]$, $\ord(A)<m$, $\ord(B)<n$, that is
$\dres(P,Q)$ belongs to the elimination ideal $(P,Q)\cap K$.

\item $\dres(P,Q)=0$ if and only if $P=\bar{P} R$, $Q=\bar{Q} R$,
with $\ord(R)>0$ and $\bar{P}, \bar{Q},R\in K[\partial]$.
\end{enumerate}
\end{prop}

Observe that Proposition \ref{prop-Propertiesdres}, 1, indicates that $AP+BQ$ is an operator of order zero, the terms in $\partial$ of degree greater than zero have been eliminated. Furthermore, Proposition \ref{prop-Propertiesdres}, 2 states that  $\dres(P,Q)=0$ is a condition on the coefficients of the operators that guarantees a right common factor.

\para

Given a fundamental system of solutions $y_1,\ldots ,y_n$ of $P=0$, let us denote by $W(y_1,\ldots ,y_n)$ the Wronskian matrix
\[
W(y_1,\ldots,y_n)=\left[\begin{array}{ccc}
y_1 & \cdots & y_n\\
\partial y_1 & \cdots & \partial y_n\\
\vdots & \vdots & \vdots\\
\partial^{n-1} y_1 & \cdots & \partial^{n-1} y_n
\end{array}\right]
\]
and by $w(y_1,\ldots ,y_n)$ its determinant.  As in the case of the classical algebraic resultant there is a Poisson formula for $\dres(P,Q)$.

\begin{prop}[\cite{Ch}, Theorem 5, see also \cite{Prev}]\label{prop-Poisson}
Given $P,Q\in K[\partial]$ with respective orders $n$ and $m$, leading coefficients $a_n$ and $b_m$ and fundamental systems of solutions $y_1,\ldots ,y_n$ and $z_1,\ldots ,z_m$ respectively of $P=0$ and $Q=0$. It holds,
\begin{align*}
\dres(P,Q)&=(-1)^{nm} a_n^m\frac{w(Q(y_1),\ldots ,Q(y_n))}{w(y_1,\ldots ,y_n)}=b_m^n\frac{w(P(z_1),\ldots ,P(z_m))}{w(z_1,\ldots ,z_m)}.\\
\end{align*}
\end{prop}

\subsubsection{Subresultant sequence}\label{subsec-subres}

We introduce next the subresultant sequence for $P$ and $Q$, which was defined in \cite{Ch}, see also \cite{Li}.
For $k=0,1,\ldots ,N:=\min\{n,m\}-1$ we define the matrix $S_k(P,Q)$ to be the coefficient matrix of the extended system of differential operator
\[\Xi_k(P,Q)=\{\partial^{m-1-k} P,\ldots \partial P, P, \partial^{n-1-k}Q,\ldots ,\partial Q, Q\}.\]
Observe that $S_k(P,Q)$ is a matrix with $n+m-2k$ rows, $n+m-k$ columns and entries in $K$.
For $i=0,\dots ,k$ let $S_k^i(P,Q)$ be the squared matrix of size $n+m-2k$ obtained by removing the columns of $S_k(P,Q)$ indexed by $\partial^{k},\ldots ,\partial,1$, except for the column indexed by $\partial^{i}$. Whenever there is no room for confusion we denote $S_k(P,Q)$ and $S_k^i(P,Q)$ simply by $S_k$ and $S_k^i$ respectively.
The {\it subresultant sequence} of $P$ and $Q$ is the next sequence of differential operators in $K[\partial]$:
\[\cL_k=\sum_{i=0}^k \det(S_k^i) \partial^i,\,\,\, k=0,\ldots ,N.\]
In this paper we will only use $\cL_1=\det(S_1^0)+\det(S_1^1)\partial$ where
\begin{equation}\label{eq-S10}
S_1^0:={\rm submatrix}(S_1,\hat{\partial})
\end{equation}
and
\begin{equation}\label{eq-S11}
S_1^1:={\rm submatrix}(S_1,\hat{1})
\end{equation}
are the submatrices of $S_1 =S_1 (P,Q)$ obtained by removing columns indexed by $\partial$ and $1$ respectively.

Recall that $K[\partial]$ is a left Euclidean domain. If $\ord(P)\geq \ord(Q)$ then $P=qQ+r$ with $\ord(r)<\ord(Q)$, $q,r\in K[\partial]$. Let us denote by $\gcd(P,Q)$ the greatest common (right) divisor of $P$ and $Q$.

\begin{thm}[\cite{Ch}, Theorem 4. Differential Subresultant Theorem]\label{thm-subres}
Given differential operators $P$ and $Q$ in $K[\partial]$,  $\gcd(P,Q)$ is a differential operator of order $r$ if and only if:
\begin{enumerate}
\item $\cL_k$ is the zero operator for $k=0, 1, ,\ldots ,r-1$ and,

\item $\cL_r$ is nonzero.
\end{enumerate}
Then $\gcd(P,Q)=\cL_r$.
\end{thm}

\begin{rem} From \ref{thm-subres} we obtain the following consequences.
\begin{enumerate}
\item Given $\cL_r=\gcd(P,Q)$ then $P=\bar{P}\cL_r$ and $Q=\bar{Q}\cL_r$, $\bar{P},\bar{Q}\in K[\partial]$.

\item The $\gcd(P,Q)$ is nontrivial (it is not in $K$) if and only if $\cL_0=\dres(P,Q)=0$.
\end{enumerate}
\end{rem}

We will define resultants and first subresultants of KdV Schr\" odinger differential operators. Next we make some remarks in the formal case to be used later when $u$ is specialized to a potential in $K$.

\begin{rem}\label{prop-c1}
Let us consider the formal Schr\" odinger operator $L=-\partial^2+u$
and the differential operator $\hat{P}_{2s+1}(u,c^s)$ defined in
\eqref{eq-Ac2s+1}. The following statements hold:
\begin{enumerate}
\item We have the following formula:
$$\dres(L-\lambda, \hat{P}_{2s+1}-\mu)=-\mu^2+R_{2s+1}(u,c^s,\lambda)$$
where $R_{2s+1}(u,c^s,\lambda)$ is a polynomial in $C\{u\}[c^s, \lambda]$.

\item The determinant of $S_1^1(L-\lambda,\hat{P}_{2s+1}-\mu)$ is a polynomial $\varphi_2$ in $C\{u\}[c^s, \lambda]$.

\item The determinant of $S_1^0(L-\lambda, \hat{P}_{2s+1}-\mu)$ equals $-\mu-\alpha$, where $\alpha\in C\{u\}[c^s, \lambda]$ .
\end{enumerate}
\end{rem}

\subsection{Formal examples}\label{sec-formal ex}

We would like to highlight now that all definitions in Section \ref{sec-KdV} are algorithms due to Lemma \ref{lem-kdv}. Since $u$ is a differential indeterminate over $C$ and $\kdv_n$, $v_n$ are differential polynomials in $C\{u\}$, it is important to note that Lemma \ref{lem-kdv} guarantees that they are well defined, the symbolic integral $\partial^{-1}(\kdv_n)$  of $\kdv_n$ can be computed with any software for symbolic computation.

We implemented these definitions with Maple 18. The first iterations of these computations are:
\begin{align*}
&\kdv_1=-\frac{1}{4} u'''+\frac{3}{2} u u',\,\,\,
\kdv_2=\frac{1}{16} u^{(5)}-\frac{5}{8} u u'''-\frac{5}{4} u' u''+\frac{15}{8} u^2 u',\\
&\kdv_3= \frac{35}{16} u' u^3-\frac{35}{32} u^2 u'''-\frac{35}{8} u u' u''-\frac{35}{32} (u')^3+\frac{21}{32} u^{(4)} u'+\frac{7}{32} u^{(5)} u+\frac{35}{32} u''' u''-\frac{1}{64} u^{(7)} \\
&v_1=\frac{u}{2},\,\,\, v_2=\frac{3}{8}u^2-\frac{1}{8}u'',\,\,\,v_3=\frac {5\,{u}^{3}}{16}-{\frac {5\,u''u}{16}}-{\frac {5\,{(u')
}^{2}}{32}}+\frac{u^{(4)}}{32}
\end{align*}
To implement the KdV differential operators $\hat{P}_{2n+1}$ we used formula \eqref{eq-Ac2s+1} and the Maple package {\rm OreTools}. They are differential operators in $C\{u\} [\partial]$ the ring of differential polynomials in $u$. In fact by \eqref{eq-P2n+1} we obtain

\begin{align*}
P_3=&-\partial^3+\frac{3}{2} u \partial +\frac{3}{4}u',\,\,\,
    P_5=\partial^5-\frac{5}{2}u \partial^3-\frac{15}{4} u'\partial^2+\frac{15}{8}u^2\partial-\frac{25}{8}u''\partial -\frac{15}{16}u'''+\frac{15}{8}u u',\\
    P_7=&{\frac {105\,{u}^{2}u'}{32}}-{\frac {105\,u' u''}{16}}-{
\frac {105\,u'''u}{32}}+{\frac {63\,u^{(5)}}{64}}+ \left( {\frac {
35\,{u}^{3}}{16}}-{\frac {245\,{u'}^{2}}{32}}-{\frac {175\,u'' u}{16}}+{\frac {161\,u^{(4)}}{32}} \right) \partial\\
&+ \left( {\frac {175\,u'''}{16}}-{\frac {105\,u'u}{8}} \right) {\partial}^{2}+ \left( {\frac {
105\,u''}{8}}-{\frac {35\,{u}^{2}}{8}} \right) {\partial}^{3}+{\frac {35
\,u'{\partial}^{4}}{4}}+\frac{7}{2}\,u{\partial}^{5}-{\partial}^{7}.
\end{align*}

Let us consider the formal Schr\" odinger operator $L=-\partial^2+u$ and $\hat{P}_3(u,c^1)=P_3+c_1P_1$.  Let $\lambda$ and $\mu$ be algebraic indeterminates as in Section \ref{sec-notation}. The next differential resultant will be of interest

\begin{equation*}
\dres(L-\lambda,\hat{P}_3-\mu)=-\mu^2 +R_{3}(u,c^1,\lambda)=
-\mu^2-\lambda^3-2c_1 \lambda^2+p_1(u,c_1) \lambda +p_0(u,c_1)
\end{equation*}

where
\begin{equation*}
\begin{array}{ll}
p_1(u,c_1)=\frac{1}{4}u''+\frac{3}{4}u^2+c_1 u-c_1^2& \mbox{ with } \partial(p_1(u,c_1))=\KdV_1(u,c^1),\\
p_0(u,c_1)=\frac{1}{16} (u')^2 +\frac{1}{4}u^3-\frac{1}{8}u'' u-\frac{1}{4} u''c_1+u^2 c_1+u c_1^2&  \\
 & \mbox{ with } \partial(p_0(u,c_1))=\left(\frac{u}{2}+c_1\right) \KdV_1(u,c^1).
\end{array}
\end{equation*}

The differential subresultant of $L-\lambda$ and $\hat{P}_3-\mu$ is
\[\cL_1=\det(S_1^0)+\det(S_1^1)\partial= \left(-\mu-\frac{u'}{4}\right)+\left(\frac{u}{2}+c_1+\lambda\right)\partial\]
with
\begin{equation*}
S_1^0=\left[\begin{array}{ccc}
-1&0&u'\\
0&-1&u-\lambda\\
-1&0&\frac{3}{4}u'-\mu
\end{array}\right]
\mbox{ and }
S_1^1=\left[\begin{array}{ccc}
-1&0&u-\lambda\\
0&-1&0\\
-1&0&\frac{3}{2}u+c_1
\end{array}\right].
\end{equation*}

\para

\section{Integration  constants for KdV potentials}\label{sec-constants}

\para
In this section, we specialize $u$ to a potential $\tilde{u}$ in the differential field $K$, with field of constants $C$.
First observe that  $\KdV_n(\tilde{u},c^n)$ is equal to zero if there exists a set of constants $\tilde{c}^n\in C^n$ such that $\KdV_n(\tilde{u},\tilde{c}^n)=0$.

\subsection{Flag of constants for KdV potentials}\label{subsec-flag}

\para

Having fixed a potential $\tilde{u}$ in $K$, we will study next the determination of a set of constants $\tilde{c}^n$ sa\-tisfying the equation $\KdV_n(\tilde{u},c^n)=0$, $n\in\bbN$, in the set of algebraic variables $c^n=(c_1,\ldots ,c_n)$. We will explore the structure of the sets of constants verifying the KdV equations for a given potential $\tilde{u}$. Our method was motivated by \cite{GH}, Remark 1.5, where the problem is posted but no computational solution is given. We addressed the problem with the goal of giving an algorithm for the computation of constants, that is included in Section \ref{sec-alg_constants}.

\para

Recall that $\kdv_n$ is the differential polynomial in $C\{u\}$ given by \eqref{eq-kdv}. After replacing $u=\tilde{u}$ in $\kdv_n$ we obtain an element of $K$ denoted by $k_n=\kdv_n(\tilde{u})$.
Observe that the linear equation in $c_1,\ldots, c_n$
\begin{equation}\label{eq-KdVus}
\KdV_n(\tilde{u},c^n)=\kdv_n(\tilde{u})+\sum_{\ell=0}^{n-1} \kdv_\ell (\tilde{u}) c_{n-\ell }=k_n+k_{n-1}c_1+\ldots k_1 c_{n-1}+k_0 c_n=0
\end{equation}
determines an affine hyperplane in $K^n$. Let $\cH_n$ be its intersection with $C^n$
\[\cH_n:=\{\xi\in C^n\mid \KdV_n(\tilde{u},\xi)=0\}.\]

\begin{defi}\label{def-KdVpot}
We call a potential $\tilde{u}$ in a differential field $K$, a {\rm KdV potential} if
there exists $n\geq 1$ such that $\cH_n\neq \emptyset$.
Let $s$ be the smallest positive integer such that $\cH_s\neq \emptyset$,
we call $s$ the {\rm KdV level} of $\tilde{u}$. We will write $u_s$ for a  KdV potential  $\tilde{u}$ of KdV level $s$.
\end{defi}

Thus the level $s$ of a potential indicates the first equation $\KdV_s=0$ that is satisfied by $u_s$ for a given set of constants. Furthermore, the next proposition explains that $u_s$ satisfies $\KdV_n=0$ for all $n>s$. In addition, the choice of constants is unique in the first level but not in the remaining ones.

\begin{prop}\label{prop-Hn}
Given a potential $u_s$ the following statements are satisfied:
\begin{enumerate}
\item $\cH_s=\{\bar{c}^s\}$ with $\bar{c}^s=(c_1^s,\ldots ,c_s^s)\in C^s$.
\item For all $n>s$, the $C$ vector space $\cV_n:=\{\xi\in C^n\mid \KdV_n(u_s,\xi)-k_n=0\}$ has dimension $n-s$, namely
\begin{equation}\label{eq-Vn}
\cV_{s+1}=\langle (1,\bar{c}^s)\rangle,\,\,\, \cV_{n+1}=\cV_n\oplus \cW_n,\mbox{ with }\cW_n=\langle (1,\bar{c}^s,0,\ldots ,0)\rangle,
\end{equation}
identifying $\cV_n$ with its natural embedding in $\cV_{n+1}$ defined by $x\mapsto (0,x)$. Furthermore, there is an infinite flag
\begin{equation}\label{eq-flag}
\cV_s\subset\cdots\subset  \cV_n\subset\cdots,
\end{equation}
that we call the {\rm flag of constants for} $u_s$.
\item For all $n>s$, we have $\cH_{n}=\bar{c}^s_n+\cV_n$, with $\bar{c}^s_n=(c_1^s,\ldots ,c_s^s,0,\ldots ,0)\in C^n$. Furthermore, there is an infinite flag of affine spaces
\begin{equation}\label{eq-flag}
\cH_s\subset\cdots\subset  \cH_n\subset\cdots,
\end{equation}
identifying $\cH_n$ with its natural embedding in $\cH_{n+1}$ defined by $x\mapsto (x,0)$.
\end{enumerate}
\end{prop}
\begin{proof}
\begin{enumerate}
\item If there exists $\xi=(\xi_1,\ldots ,\xi_s)\neq \bar{c}^s$ in $\cH_s$ then for some $1\leq i\leq s$, $\xi_i-c_i^s\neq 0$ and
\[k_{s-i}+k_{s-i-1} \frac{\xi_{i+1}-c_{i+1}^s}{\xi_i-c_i^s}+\ldots +k_0  \frac{\xi_s-c_s^s}{\xi_i-c_i^s}=0,\]
contradicting that $\cH_{s-i}=\emptyset$.
\item By \eqref{eq-kdv} and \eqref{eq-KdV} we have
\begin{align}\label{eq-rec}
\cR^{n-s}(\KdV_s(u,c^s))&=\cR^{n-s}(\kdv_s(u))+\sum_{\ell=0}^{s-1} c_{s-\ell}\cR^{n-s}(\kdv_{\ell}(u))=\kdv_n(u)+\sum_{\ell=0}^{s-1} c_{s-\ell }\kdv_{n-s+l}(u).
\end{align}
Let us consider the recursion operator \eqref{eq-recursion} for $u=u_s$, that is $\cR_s=-\frac{1}{4}\partial^2+u_s+\frac{1}{2}u_s'\partial^{-1}$.
Replacing $u$ by $u_s$ and $c^s$ by $\bar{c}^s$ in \eqref{eq-rec} we obtain
\begin{equation}\label{eq-kdvc}
k_{n}=- k_{n-1}c_1^s-\cdots -k_{n-s}c_s^s,
\end{equation}
since $\cR_s$ is a linear operator acting on $C\langle u_s\rangle$ and $\KdV_s(u_s,\bar{c}^s)=0$.

We prove \eqref{eq-Vn} by induction on $n$. An element $\xi=(\xi_1,\ldots ,\xi_{s+1})$ of $\cV_{s+1}$ verifies $k_s\xi_1+\cdots +k_0\xi_{s+1}=0$, and taking $n=s$ in \eqref{eq-kdvc} we get
\[k_{s-1}(\xi_2-\xi_1 c_1^s)+\cdots +k_{0}(\xi_{s+1}-\xi_1 c_s^s)=0.\]
Then $\xi=\xi_1 (1,\bar{c}^s)$, because 1. implies that $\cV_s$ is the null space. Let us assume that $\cV_n$ has basis $\{w_1,\ldots ,w_{n-s}\}$. Observe that
\[\cV_{n+1}\cap \{\xi\in\ C^{n+1}\mid \xi_1=0\}=\{0\}\times \cV_n\]
has basis $\cB=\{(0,w_1),\ldots ,(0,w_{n-s})\}$. Using \eqref{eq-kdvc} we can prove that $\xi\in \cV_{n+1}$ verifies
\[k_{n-1}(\xi_2-\xi_1c_1^1)+\ldots +k_{n-s}(\xi_{s+1}-\xi_1c_s^s)+k_{n-s-1}\xi_{s+2}+\cdots +k_0\xi_{n+1}=0.\]
Let $w=(1,\bar{c}^s,0,\ldots ,0)\in C^{n+1}$, then $\xi-\xi_1w\in \{0\}\times \cV_n$, which proves that $\{w\}\cup C\{u\}$ is a basis of $\cV_{n+1}$ of size $n+1-s$.

\item Substituting \eqref{eq-kdvc} in $\KdV_{n}(u_s,c^{n})$ gives
\begin{align*}
\KdV_{n}(u_s,c^{n})=&k_n (c_1-c_1^s)+\cdots +k_{n-s}(c_s-c_s^s)+k_{n-s-1} c_{s+1}+\ldots +k_0 c_{n},
\end{align*}
proving that $\cH_{n}=\bar{c}^s_{n}+\cV_{n}$. Similarly we can prove that given $\xi\in \cH_i$, $s\leq i\leq n-1$ then
\begin{align*}
\KdV_{i+1}(u_s,c^{i+1})=&k_i (c_1-\xi_1)+\cdots +k_1(c_i-\xi_i)+k_0 c_{i+1}.
\end{align*}
Therefore $\cH_i\times\{0\}\subset \cH_{i+1}$ and \eqref{eq-flag} follows.
\end{enumerate}
\end{proof}

The previous proposition shows that the flag of constants of a KdV potential $u_s$ of KdV level $s$ is determined by $\cH_s=\{\bar{c}^s\}$. We call $\bar{c}^s$ the {\it basic constants vector} of $u_s$.

\begin{ex}
As a first example, let us consider $\tilde{u}=6/x^2$ in $K=\bbC(x)$. One can easily check that $\tilde{u}$ is a KdV potential of level 2. It does not verify $\KdV_1(u,c_1)=kdv_1(u)+c_1 kdv_0(u)=0$ for any $c_1\in\bbC$ but $\KdV_2(\tilde{u},(0,0))=kdv_2(\tilde{u})=0$ and  $\KdV_n(u,c^n)=0$, $n> 2$ is satisfied by $u=\tilde{u}$ for infinitely many choices of $c^n\in \bbC^n$. Its basic constant vector is $\bar{c}^2=(0,0)$. More examples can be found in Section \ref{sec-Ejemplos}.
\end{ex}

The computation of the basic constants vector is algorithmic at least for a big family of potentials, as we explain in the next section.

\subsection{Computing the integration constants}\label{sec-alg_constants}

We designed an algorithm that decides if a potential $\tilde{u}$ in $K$ is a KdV potential and returns its level $s$ and basic constants vector $\bar{c}^s$. For this purpose we restrict to the case of potentials $\tilde{u}$ that are rational functions in an element $\eta$ in $K$.
 Let $C(\eta)$ be the field of rational functions in $\eta$, we assume that $\tilde{u}\in C(\eta)$. Furthermore we assume that $(\eta')^2\in C(\eta)$. This request is necessary for a hamiltonian algebraization  in order to preserve the Galoisian behavior of the factors, see \cite{AMW}. This situation is satisfied by the three families of KdV potentials that we will use to illustrate all the results of this paper in Section \ref{sec-Ejemplos}.

\para

We can distinguish two cases. If $\eta'\in C(\eta)$ then $C\langle \tilde{u}\rangle\subset C(\eta)$.
If $(\eta')^2\in C(\eta)$ but $\eta'\notin C(\eta)$ then $C\langle \tilde{u}\rangle$ is contained in the linear space $V=C(\eta)\oplus \eta'C(\eta)$.

\begin{lem}\label{lem-aV}
Let us consider $a\in V$.
\begin{enumerate}
\item If $a\in  \eta'C(\eta)$ then $a'\in C(\eta)$.
\item $a\in C(\eta)$ if and only if $a'\in  \eta'C(\eta)$.
\end{enumerate}
\end{lem}
\begin{proof}
Clearly, if $a\in  \eta'C(\eta)$ then $a'\in C(\eta)$ and also if $a\in C(\eta)$ then $a'\in  \eta'C(\eta)$. Let us assume $a'\in  \eta'C(\eta)$, with $a=h_0+h_1 \eta'$, $h_0,h_1\in C(\eta)$ and $h_1\neq 0$.
Then $\partial (a-h_0)=\partial(\eta'h_1)\in C(\eta)\cap  \eta'C(\eta)=0$ thus $\eta'h_1$ is a constant which contradict that $\eta'\notin C(\eta)$.
\end{proof}

Recall that $v_n$ is the differential polynomial in $C\{u\}$ given by \eqref{eq-fn}. After replacing $u=\tilde{u}$ in $v_n$ we obtain an element of $K$ that will be denoted by $v_n(\tilde{u})$. If $\eta'\in C(\eta)$ then $kdv_n(\tilde{u})\in C(\eta)$ and also $v_n(\tilde{u})\in C(\eta)$, for all $n\in \bbN$.

\begin{lem}
Let us consider a potential $\tilde{u}\in C(\eta)$.
Then $\kdv_n(\tilde{u})\in \eta'C(\eta)$ and $v_n(\tilde{u})\in C(\eta)$, for all $n$.
\end{lem}
\begin{proof}
Given $\tilde{u}\in C(\eta)$ we can easily prove that $\partial^{n}(\tilde{u})$ belongs to $\eta'C(\eta)$  for $n$ odd and belongs to $C(\eta)$ for $n$ even. Observe that $v_2(\tilde{u})\in C(\eta)$ (see Section \ref{sec-formal ex}) and thus $\kdv_1(\tilde{u})\in \eta'C(\eta)$ by Lemma \ref{lem-aV}, since $\kdv_1=\partial(v_2)$. Let us assume that $\kdv_n(\tilde{u})=\partial(v_{n+1})(\tilde{u})\in \eta'C(\eta)$, then Lemma \ref{lem-aV} implies $v_{n+1}\in C(\eta)$. Since $\kdv_{n+1}=\cR(\kdv_n)$ we have
\[\kdv_{n+1}(\tilde{u})=-\frac{1}{4}\partial^2(\kdv_n(\tilde{u}))+\tilde{u}\kdv_n(\tilde{u})+\frac{1}{2}\tilde{u}' v_{n+1}(\tilde{u}),\]
which is the sum of terms in $\eta'C(\eta)$, hence $\kdv_{n+1}(\tilde{u})\in \eta'C(\eta)$.
\end{proof}

From the previous lemmas and \eqref{eq-KdVus}, the next result follows.

\begin{prop}\label{prop-pnqn}
Given $\tilde{u}\in C(\eta)$ then
\begin{equation}
\KdV_n(\tilde{u},c^n)=
\left\{
\begin{array}{ll}
\frac{p_n(\eta)}{q_n(\eta)}&\mbox{ if }\eta'\in C(\eta),\\
\eta'\frac{p_n(\eta)}{q_n(\eta)}&\mbox{ if }(\eta')^2\in C(\eta),\,\,\,\eta'\notin C(\eta),
\end{array}
\right.
\end{equation}
where $p_n=\sum_d l_d \eta^d$ and $q_n\in C[\eta]$ are polynomials in $\eta$, with $l_d$ linear expressions in $c_1,\ldots ,c_n$ over $C$.
\end{prop}

We are ready to give the announced algorithm.

\begin{alg}\label{alg-constants} (Basic Constants Vector)
Let $\eta\in K$ be such that $(\eta')^2\in C(\eta)$.
\begin{itemize}
\item \underline{\sf Given} $\tilde{u}\in C(\eta)$ and $s^*\geq 1$.
\item \underline{\sf Decide} if $\tilde{u}$ is a KdV potential of KdV level smaller than or equal to $s^*$ and \underline{\sf return} the KdV level $s$ and its basic constants vector $\bar{c}^s$.
\end{itemize}
\begin{enumerate}
\item Set $n:=1$.
\item Replace $u$ by $\tilde{u}$ in $\KdV_n(u,c^n)$ to obtain $\frac{p_n(\eta)}{q_n(\eta)}$, as in Proposition \ref{prop-pnqn}.
\item Collect the coefficients in $\eta$ of $p_n$ to obtain a nonhomogeneous system $\cS_n$
of linear equations over $C$ in the unknowns $c_1,\ldots c_n$.
\item If $\cS_n$ has a solution $\xi\in C^n$, {\sf return} $s:=n$ and $\bar{c}^s:=\xi$.
\item If $n=s^*$ {\sf return \textquotedblright it is not a KdV potential up to the required level\textquotedblright  }.
\item Set $n:=n+1$ and go to Step 2.
\end{enumerate}
\end{alg}

See examples in Sections \ref{sec-RM} and \ref{sec-Elliptic}.

\section{Spectral curves of KdV Schr\" odinger operators}\label{sec-spectralcurves}

In this section we will study the centralizers of Schr\" odinger operators $L_s=L(u_s)=-\partial^2+u_s$, where $u_s$ is a KdV potential of KdV level $s$ and basics constant vector $\bar{c}^s \in \mathbb{C}^s$, as defined in Section \ref{subsec-flag}. We will call $L_s$ a {\it KdV Schr\" odinger operator}  or  $\KdV_s$ for short. This will allow us to define the spectral curve of $L_s$ by means of an operator of  order $2s+1$ commuting with $L_s$. For this purpose we need to study the centralizer of the  operator $ L_s $. We will show that this centralizer is generated by $ L_s $ and another operator $ A_{2s + 1} $. This pair, $\{ L_s, A_{2s + 1} \} $, will be the one we use to calculate an equation of the spectral curve associated with $ L_s $.

\para

\subsection{Centralizers and Burchnall-Chaundy polynomials}\label{subsec-Centralizers}

\para

To start we summarize some results from \cite{Good} about centralizers of differential operators.
Let $P=a_n \partial^n+\cdots +a_1\partial+a_0$ be an operator in $E[\partial]$.
Let us denote by $\cC_E(P)$ the centralizer of $P$ in $E[\partial]$, that is
\[\cC_E(P)=\{Q\in E[\partial]\mid PQ=QP\}.\]
By \cite{Good}, Theorem 4.1,
if $n$ and $a_n$ are non zero divisors in $E$ then $C_E(P)$ is commutative.
Let $\cC^{\infty}$ be the ring of infinitely-many times differentiable complex-valued functions on the real line.
By \cite{Good}, Corollary 4.4, $\cC_{\cC^{\infty}}(P)$ is commutative if and only if there is no nonempty open interval on the real line on which the functions $\partial(a_0),a_1,\ldots ,a_n$ all vanish.

\para

Details of the evolution of these results from various previous works are given in \cite{Good}. We chose this reference because it simplifies the existing methods and applies them in as wide a context as reasonable. Precursors of the commutativity results are Schur \cite{Sch}, Flanders \cite{Fl}, Krichever \cite{K77}, Amitsur \cite{Amit}, Carlson and Goodearl \cite{CG}. Results describing centralizers $\cC_R(P)$ as a free module of finite rank appear in \cite{Fl}, \cite{Amit}, \cite{CG} and in Ore's well known paper \cite{Ore}.

\para

Recall that a commutative ring $E$ is called reduced if it has no nonzero nilpotent element. Observe that $\cC^{\infty}$ is not a field, but it is a reduced ring whose ring of constants is the field $\bbC$.

\begin{thm}\label{thm-good}
Let $E$ be a reduced differential ring whose subring $F$ of constants is a field. Let us assume that $n$ is invertible in $F$ and $a_n$ is invertible in $E$.
\begin{enumerate}
\item (\cite{Good}, Theorem 4.2) $\cC_E(P)$ is a commutative integral domain.

\item (\cite{Good}, Theorem 1.2) Let $X$ be the set of those $i$ in $\{0,1,2,\ldots ,n-1\}$ for which $\cC_E(P)$ contains an operator of order congruent to $i$ module $n$. For each $i\in X$ choose $Q_i$ such that $\ord(Q_i)\equiv i ({\rm mod}\, n)$ and $Q_i$ has minimal order for this property (in particular $0\in X$, and $Q_0=1$). Then $\cC_E(P)$ is a free $F[P]$-module with basis $\{Q_i\mid i\in X\}$. Moreover, the rank $t$ of $\cC_E(P)$ as a free $F[P]$-module is a divisor of $n$.
\end{enumerate}
\end{thm}

We are ready now to describe the centralizer $\cC_K(L_s)$ in $K[\partial]$ of the KdV Schr\" odinger operator $L_s$. We do so by generalizing an example in \cite{Good}, Section 1.2. In addition, by \cite{CG}, Theorem 1.6 we know that $\cC_{K}(L_s)$ has rank $2$ as a free $C[L_s]$-module.

\para

Replacing $u$ by $u_s$ and $c^n$ by $\bar{c}_n^s=(\bar{c}^s,0,\ldots ,0)$ in the family of KdV differential operators $\hat{P}_{2n+1}(u,c^n)$ defined in \eqref{eq-Ac2s+1}, we obtain a family of differential operators in $K[\partial]$
\begin{equation}\label{eq-AA2n+1}
A_{2n+1}:=\hat{P}_{2n+1}(u_s,\bar{c}_n^s),\mbox{ for all }n\geq s.
\end{equation}
As a consequence of \eqref{eq-PLKdV} and Proposition \ref{prop-Hn} we have
\begin{equation}\label{eq-Laxcond}
[A_{2n+1},L_s]=\KdV_n(u_s,\bar{c}_n^s)=0,\mbox{ for all }n\geq s.
\end{equation}
Thus $A_{2n+1}\in \cC_K(L_s)$, for all $n\geq s$. The next result shows that $A_{2s+1}$ has an important role in the description of the centralizer of $L_s$, it is the {\it differential operator that determines the centralizer of} $L_s$.

\begin{thm}\label{thm-centralizer}
Let $L_s$ be a KdV Schr\" odinger operator.
The centralizer of $L_s$ in $K[\partial]$ equals the free $C[L_s]$-module of rank $2$ with basis $\{1, A_{2s+1}\}$, that is
\[\cC_K(L_s)=\{p_0(L_s)+p_1(L_s)A_{2s+1}\mid p_0,p_1\in C[L_s]\}=C[L_s]\langle 1,A_{2s+1}\rangle.\]
\end{thm}
\begin{proof}
We will prove that there does not exist an operator of odd order smaller that $2s+1$ in $\cC_K (L_s)$. By Theorem \ref{thm-good}, 2, this implies that $\cC_K (L_s)=C[L_s]\langle 1,A_{2s+1}\rangle$.

Let us consider a monic differential operator $Q\in K[\partial]$ of order $2n+1$ with $n<s$. Let $P_{2n+1}(u)$ be the family of KdV differential operators defined in \eqref{eq-A2s+1} and denote by $P_{2n+1}^s:=P_{2n+1}(u_s)$. Since $\{P_{2i+1}^s\}_{i\leq n}$ and $\{L_s^i\}_{i\leq n}$ are families of operators in $K[\partial]$ of odd and even orders less than $2n+1$ respectively, we divide $Q$ by those families and write
\[Q=\sum_{i=0}^n q_{2i+1} P_{2i+1}^s+\sum_{i=0}^n q_{2i} L_s^i\]
with $q_{2n+1}=1$ and $q_{2i+1},q_{2i}\in K$.
To compute $[Q,L_s]$, observe that $[a,L_s]=\partial^2 (a)+2\partial (a)\partial$, for $a\in K$ and
\[[q_{2i+1} P_{2i+1}^s,L_s]=-\partial^2(q_{2i+1})P_{2i+1}^s-2\partial(q_{2i+1})\partial P_{2i+1}^s+q_{2i+1}\kdv_i(u_s)\]
and
\[[q_{2i}L_s^i,L_s]=[q_{2i},L_s]L_s^i=(\partial^2(q_{2i})+2\partial(q_{2i})\partial)L_s^i.\]
Thus in $[Q,L_s]$ the only term of order $2i+2$ is the leading term of $\partial P_{2i+1}^s$ and the only term of order $2i+1$ is the leading term of $\partial L_s^i$.
If $[L_s,Q]=0$ then $\partial(q_{2i})=0$ and $\partial(q_{2i+1})=0$. Therefore $[q_{2i}L_s^i,L_s]=0$ and $q_{2i+1}\in C$, $i=0,\ldots ,n$ implies that
\[0=[Q,L_s]=\sum_{i=0}^n q_{2i+1} \kdv_i(u_s)\]
contradicting that $u_s$ has KdV level $s$. We conclude that $Q\notin \cC_K(L_s)$, which proves the result.
\end{proof}

\para

A polynomial $f(\lambda,\mu)$ with constant coefficients satisfied by a commuting pair of differential operators $P$ and $Q$ is called a {\it Burchnall-Chaundy  (BC) polynomial} of $P$ and $Q$, since the first result of this sort appeared is the $1923$ paper \cite{BC} by Burchnall and Chaundy. Generalizations (more general rings $E$) were later studied in \cite{K77}, \cite{CG} and \cite{Richter}. The next result shows that associated to the centralizer of a differential operator $P$ there are as many BC polynomials as operators in the centralizer. We will compute these polynomials using differential resultants, as it will be explained in Section \ref{sec-Differential Resultant}.

\begin{thm}
(\cite{Good}, Theorem 1.13) Let $E$ be a reduced differential ring whose subring $F$ of constants is a field.
Given any operator $Q\in C_E(P)$ there exist polynomials $p_0(P)$,$\ldots$ ,$p_{t-1}(P)$$\in F[P]$ such that
\[p_0(P)+p_1(P)Q+\cdots +p_{t-1}(P)Q^{t-1}+Q^t=0.\]
That is, there exists a nonzero polynomial $f_Q(\lambda,\mu)\in F[\lambda,\mu]$ such that $f_Q(P,Q)=0$.
\end{thm}


\subsection{Computing spectral curves}\label{sec-Spectral curves}

The relation between Burchnall and Chaundy polynomials (see Section \ref{subsec-Centralizers}) and differential resultants was given by E. Previato in \cite{Prev}. Next, we state Previato's theorem in the general case of differential operators in $K[\partial]$ (\ref{thm-Prev}) and we give an alternative proof using the Poisson formula for the differential resultant (Proposition \ref{prop-Poisson}). Then we will compute BC polynomials of KdV Schr\" odinger operators. We we will apply Previato's Theorem \ref{thm-Prev} to the computation of the spectral curve of the Lax pair $\{L_s,A_{2s+1}\}$, showing the algebraic structure of the irreducible polynomials  $f_s(\lambda,\mu)$ defining the spectral curve $\Gamma_s$.

\para

First observe that whenever the operators $P-\lambda$ and $Q-\mu$ have coefficients in the differential ring $(K[\lambda,\mu],\partial)$ (see Section \ref{sec-notation} ),
by means of the differential resultant, Proposition \ref{prop-Propertiesdres}, 1, it is ensured that we compute a nonzero polynomial,
\begin{equation}\label{eq-dresnonzero}
\dres(P-\lambda,Q-\mu)=a_n^m\mu^n-b_m^n\lambda^m+\cdots
\end{equation}
in the elimination ideal $(P-\lambda,Q-\mu)\cap K[\lambda,\mu]$. The next result implies that if $P$ and $Q$ commute then
\[\dres(P-\lambda,Q-\mu)\in (P-\lambda,Q-\mu)\cap C[\lambda,\mu].\]

\begin{thm}[E. Previato, \cite{Prev}]\label{thm-Prev}
Given $P,Q\in K[\partial]$ such that $[P,Q]=0$ then
$$g(\lambda,\mu)=\dres(P-\lambda,Q-\mu)\in C[\lambda,\mu]$$
and also $g(P,Q)=0$.
\end{thm}
\begin{proof}
Let $y_1,\ldots ,y_n$ be a fundamental system of solutions of $(P-\lambda)(Y)=0$. Since $0=[P,Q]=[P-\lambda,Q-\mu]$ we have $(P-\lambda)(Q-\mu)(y_i)=(Q-\mu)(P-\lambda)(y_i)=0$ then $(Q-\mu)(y_i)$, $i=1,\ldots ,n$ are solutions of $(P-\lambda)(Y)=0$. Then, there exists a matrix $M$ with entries in the algebraic closure  $\mathfrak{C}$ of $C(\lambda,\mu)$ such that there exists a matrix $M$ with entries in $\mathfrak{C}$ such that
\[W((Q-\mu)(y_1),\ldots ,(Q-\mu)(y_n))=W(y_1,\ldots,y_n) M.\]
By Proposition \ref{prop-Poisson},
\[\dres(P-\lambda,Q-\mu)=\frac{w((Q-\mu)(y_1),\ldots ,(Q-\mu)(y_n))}{w(y_1,\ldots ,y_n)}=\frac{w(y_1,\ldots,y_n) \det(M)}{w(y_1,\ldots ,y_n)}=\det(M),\]
which belongs to $K[\lambda,\mu]\cap \mathfrak{C}=C[\lambda,\mu]$.

The last statement of this theorem follows from the fact that $g(\lambda,\mu)=\dres(P-\lambda,Q-\mu)$ belongs to the differential ideal generated by $P-\lambda$ and $Q-\mu$ in $K[\lambda,\mu][\partial]$. Therefore
\[g(\lambda,\mu)=A(P-\lambda)+B(Q-\mu), \mbox{ with }A,B\in K[\lambda,\mu][\partial].\]
Since $P$ and $Q$ commute then $g(P,Q)=0$.
\end{proof}

The previous theorem shows that BC polynomials (defined in Section \ref{subsec-Centralizers}) can be computed using differential resultants.
Let us suppose that $[P,Q]=0$ and let $f(\lambda,\mu)$ be the square free part of $\dres(P-\lambda,Q-\mu)\in C[\lambda,\mu]$ (i.e. the product of the different irreducible components of $g$). The affine plane algebraic curve defined by $f$
\begin{equation}
\Gamma:=\{(\lambda,\mu)\in C^2\mid f(\lambda,\mu)=0\}
\end{equation}
is known as the {\it spectral curve of the pair $\{P,Q\}$}.

\para

Let us suppose that $f(\lambda,\mu)$ is an irreducible polynomial in $K[\lambda,\mu]$ and denote by $(f)$ the prime ideal generated by $f$ in $K[\lambda,\mu]$. As a polynomial in $C[\lambda,\mu]$ is also irreducible and the ideal generated by $f$ in $C[\lambda,\mu]$ is also prime, abusing the notation we will also denote it by $(f)$ and distinguish it by the context. Let us denote by $C(\Gamma)$ and $K(\Gamma)$ the fraction fields of the domains $C[\lambda,\mu]/ (f)$ and $K[\lambda,\mu]/(f)$ respectively. Observe that $C(\Gamma)$ and $K(\Gamma)$ are usually interpreted as rational function on $\Gamma$.

\begin{rem}\label{rem-Importante}
As differential operators in $K[\lambda,\mu][\partial]$, the operators $P-\lambda$ and $Q-\mu$ have no common nontrivial solution, see \eqref{eq-dresnonzero}, but as elements of $K(\Gamma)[\partial]$ they have a common non constant factor. By Theorem \ref{thm-subres} the first nonzero subresultant  $\cL_r=\gcrd(P-\lambda,Q-\mu)$ is the greatest common divisor of $P-\lambda$ and $Q-\mu$ in $K(\Gamma)[\partial]$. We will use subresultants in Section \ref{sec-KdV factors} to compute factorizations of KdV Schr\" odinger operators.
\end{rem}

\para

Let us consider the KdV Schr\" odinger operators $L_s=L(u_s)=-\partial^2+u_s$, where $u_s$ is a KdV potential of KdV level $s$ and basic constants vector $\bar{c}^s$, as defined in Section \ref{subsec-flag}. Let $A_{2s+1}$ be the differential operator that determines the centralizer of $L_s$, see \eqref{eq-AA2n+1}.

\begin{cor}\label{cor-Prev}
The spectral curve $\Gamma_s$ of the pair $\{L_s,A_{2s+1}\}$ is defined by the polynomial in $C[\lambda,\mu]$,
$$f_s(\lambda,\mu):=\dres(L_s-\lambda,A_{2s+1}-\mu)=-\mu^2-R_{2s+1}(\lambda),$$
where $R_{2s+1}(\lambda)$ is a polynomial of degree $2s+1$ in $C[\lambda]$. The polynomial $f_s(\lambda,\mu)$ is irreducible in $K[\lambda,\mu]$. In addition, the coefficients of $R_{2s+1}(\lambda)(u,\bar{c}^s,\lambda)$ in Proposition \ref{prop-c1} are first integrals of $KdV_s(u,\bar{c}^s)$.
\end{cor}
\begin{proof}
By \eqref{eq-Laxcond}, $[A_{2n+1},L_s]=\KdV_n(u_s,\bar{c}_n^s)=0$. Thus by Theorem \ref{thm-Prev}, $f_s\in C[\lambda,\mu]$. In addition, by Remark \ref{prop-c1}, $f_s=-\mu^2-R_{2s+1}(\lambda)$, which can be easily proved to be irreducible in $K[\lambda,\mu]$ because it has odd degree in $\lambda$.
\end{proof}

\begin{defi}
{\rm The spectral curve of $L_s$} is defined as the plain algebraic curve $\Gamma_s$ given by $f_s(\lambda,\mu)= 0$, with $f_s$ as defined in Corollary \ref{cor-Prev}.
\end{defi}

\section{ Factors of KdV Schr\" odinger operators  over spectral curves}\label{sec-KdV factors}

Let $u_s$ be a KdV potential of KdV level $s$ and basic  constants vector $\bar{c}^s$. Let $A_{2s+1}$ be the differential operator that determines the centralizer of the KdV Schr\" odinger $L_s=-\partial^2+u_s$ as in Theorem \ref{thm-centralizer}. By Corollary \ref{cor-Prev}, the spectral curve $\Gamma_s$ of the pair $\{L_s,A_{2s+1}\}$ is defined by the irreducible polynomial
$$f_s (\lambda,\mu)=\dres(L_s-\lambda,A_{2s+1}-\mu)=\mu^2-R_{2s+1}(\lambda)\in C[\lambda,\mu].$$
Let $C(\Gamma_s)$ and $K(\Gamma_s)$ be the fraction fields of the domains  $C[\lambda,\mu]/(f_s)$ and $K[\lambda,\mu]/(f_s)$. In this section, we explain how to factor $L_s-\lambda$ as an operator in $K(\Gamma_s)[\partial]$.

\subsection{KdV factors on $\Gamma_s$}\label{subsec-KdVfactors}

In this section, we consider the operators $L_s-\lambda$ and $A_{2s+1}-\mu$ as elements of $K(\Gamma_s)[\partial]$.
Let $\cL_1=\varphi_2 \partial+\varphi_1$ be the subresultant of $L_s-\lambda$ and $A_{2s+1}-\mu$ as in Section \ref{subsec-subres}.

\begin{thm}\label{thm-L1}
The greatest common factor of the differential operators $L_s-\lambda$ and $A_{2s+1}-\mu$ in $K(\Gamma_s)[\partial]$ is the order one operator $\cL_1$.
\end{thm}
\begin{proof}
Since $\cL_0=\dres(L_s-\lambda,A_{2s+1}-\mu)$ is zero in $K(\Gamma_s)$ by Theorem \ref{thm-subres} the result follows.
\end{proof}
We can take the monic greatest common factor of $L_s-\lambda$ and $A_{2s+1}-\mu$ to be
\[\partial-\phi_s,\mbox{ where }
\phi_s=-\frac{\varphi_1}{\varphi_2}.\]
The fact that $\partial-\phi_s$ is a right factor implies that
\[L_s-\lambda=(-\partial-\phi_s)(\partial-\phi_s),\mbox{ in }K(\Gamma_s)[\partial]\]
and moreover $\phi_s$ is a solution of the Ricatti equation associated to the Sch\" odinger operator $L_s-\lambda$
\begin{equation}\label{eq-Ricatti}
\partial(\phi)+\phi^2=u_s-\lambda
\end{equation}
on the spectral curve $\Gamma_s$. Therefore, we compute a solution of \eqref{eq-Ricatti} by means of the differential subresultant $\cL_1$. We will give next some details about $\phi_s$.

\begin{lem}\label{lem-phi}
The following formula holds
\begin{equation}\label{eq-phis}
\phi_s=\frac{\mu+\alpha(\lambda)}{\varphi(\lambda)},
\end{equation}
where $\alpha$ and $\varphi$ are nonzero polynomials in $K[\lambda]$. Moreover $\phi_s$ is nonzero in $K(\Gamma_s)$.
\end{lem}
\begin{proof}
By \eqref{eq-S10},\eqref{eq-S11}, we have that $\varphi_1=det(S_1^0(L_s-\lambda, A_{2s+1}-\mu))$ and $\varphi_2=det(S_1^1(L_s-\lambda, A_{2s+1}-\mu))$.
Now by Remark \ref{prop-c1}, $\varphi_1=-\mu-\alpha$ and $\alpha$, $\varphi=\varphi_2$ are nonzero polynomials in $K[\lambda]$. Observe that $\phi_s=0$ in $K(\Gamma_s)$ if and only if $
\mu+\alpha+(f_s)=0$ in $K[\Gamma_s]$. But this is not possible since $f_s$, which has degree $2$ in $\mu$, is not a factor of $\mu+\alpha$ in $K[\lambda,\mu]$. This proves the last claim.
\end{proof}

To keep notation as simple as possible, we will also write $\phi_s$ to denote the element $\phi_s$ in $K(\Gamma_s)$.
The next algorithm takes as an input a KdV potential to obtain the factor $\partial-\phi_s$ of $L_s-\lambda$ in $K(\Gamma_s)[\partial]$.

\begin{alg}\label{alg-fact} (Factorization)
\begin{itemize}
\item \underline{\sf Given} $u_s$ a KdV potential of KdV level $s$ and given $\bar{c}^s$ the basic constants vector of $u_s$.
\item \underline{\sf Return} the defining polynomial $f_s$ of the spectral curve $\Gamma_s$ and the monic greatest common divisor $\partial-\phi_s$ of $L_s-\lambda$ and $A_{2s+1}-\mu$ in $K(\Gamma_s)[\partial]$.
\end{itemize}
\begin{enumerate}
\item Define $L_s:=-\partial^2+u_s$ and $A_{2s+1}:=\hat{P}_{2s+1}(u_s,\bar{c}^s)$ as in \eqref{eq-AA2n+1}.
\item Compute $f_s=\dres(L_s-\lambda,A_{2s+1}-\mu)$, the defining polynomial of the spectral curve $\Gamma_s$.
\item Compute $\cL_1=\varphi_1+\varphi_2 \partial$, the subresultant of $L_s-\lambda$ and $A_{2s+1}-\mu$ as in Section \ref{subsec-subres}.
\item Define $\phi_s:=-\frac{\varphi_1}{\varphi_2}$.
\item Return $f_s$ and $\partial-\phi_s$.
\end{enumerate}
\end{alg}

\begin{rem}
Observe that we are computing $\phi_s$ in closed form as the quotient of two determinants $-\varphi_1/\varphi_2$, which is a well defined function over the spectral curve. As far as we know, there were no algorithms to obtain the factors $\partial-\phi_s$ of $L_s-\lambda$ over the spectral curve. In \cite{GH}, there
are  some differential recursive expressions for the factorization of $L_s-\lambda_0$ for each  point $(\lambda_0,\mu_0)$ of $\Gamma_s$. Note that our factorization algorithm is defined in $K(\Gamma_s)[\partial]$.
\end{rem}

We would like to obtain a univariate expression of $\phi_s$ using a parametric representation of $\Gamma_s$, whenever it is possible. This will allow us to give a functional representation of $\phi_s$ and as a byproduct, we will obtain a domain of definition of the solutions of $L_s-\lambda$, see Sections \ref{sec-param}, and \ref{sec-specialization}.

\subsection{Factorization for parametrizable spectral curves}\label{sec-param}

Once the factorization problem over $K(\Gamma_s)$ has been solved, in Algorithm \ref{alg-fact}, what remains is to replace $(\lambda,\mu)$ by a parametric representation $(\chi_1(\tau),\chi_2(\tau))$ of $\Gamma_s$. We are not aware of  a previous work where a global treatment of the factorization is achieved. This procedure strongly depends on the genus of the algebraic curve $\Gamma_s$. We summarize next what are the parametrization possibilities (as far as we know) and emphasize on the algorithmic aspects of the process.

\para
A key point to have a one-parameter form factorization algorithm is to obtain a global parametrization of the spectral curve. How complicated is to obtain a global parametrization depends on the genus of the curve. There are algorithms to compute the genus of an algebraic curve \cite{SWP}. In the case of rational curves there are algorithms to obtain a global parametrization \cite{SWP}. For elliptic curves we can define a meromorphic parametrization by means of the Weierstrass $\wp$-function. For all other cases, as far as we know, there are no algorithms to obtain a global parametrization. An affine algebraic curve $\Gamma$ admits at any point $P\in\Gamma$ a local parametrization in the field of Puiseux series, see for instance \cite{SWP}, Section 2.5 but in this paper we would like to talk only about the global treatment of the curve.

If an affine algebraic curve $\Gamma$ in $\bbC^2$ is rational (has genus zero) then $\Gamma$ can be parametrized by rational functions.
Let $\aleph(\tau)=\left(\chi_1(\tau),\chi_2(\tau)\right)$ in $C(\tau)^2$ be a (global) parametrization of $\Gamma$, that is:
\begin{enumerate}
\item For all $\tau_0\in C$, but a finite number of exceptions, the point $(\chi_1(\tau_0),\chi_2(\tau_0))$ is on $\Gamma$, and
\item for all $(\lambda_0,\mu_0)\in \Gamma$, but a finite number of exceptions, there exists $\tau_0\in C$ such that $(\lambda_0,\mu_0)=(\chi_1(\tau_0),\chi_2(\tau_0))$.
\end{enumerate}
A rational parametrization $\aleph (\tau)$ of $\Gamma$ gives an isomorphism from $C(\Gamma)$ to the field of rational functions $C(\tau)$, see \cite{SWP}, Section 4.1.
This can be extended to an isomorphism $K(\Gamma)\simeq K(\tau)$. More over, $K(\tau)$ is isomorphic to the fraction field $\cF=K(\chi_1(\tau),\chi_2(\tau))$ of the polynomial ring $K[\chi_1(\tau),\chi_2(\tau)]$. Since $\tau$ is an algebraic indeterminate over $K$, by condition 2, it is natural to assume that $\partial(\chi_1(\tau))=0$ and $\partial(\chi_2(\tau))=0$, which allows to extend the derivation $\partial$ of $K$ to have a differential field $(\cF,\partial)$.

We define $\tilde{\phi}_s:=\rho(\phi_s)$. Observe that $\tilde{\phi}_s$ is a nonzero element of $\cF$ since by Lemma \ref{lem-phi} $\phi_s$ is nonzero in $K(\Gamma_s)$.
We have naturally an isomorphism $\rho$ between the rings of differential operators ${\varrho}: K(\Gamma_s)[\partial]\longrightarrow\cF[\partial]$ as follows:
\[{\varrho}\left(\sum_{j}a_j\partial^j\right)=\sum_{j}\rho(a_j)\partial^j.\]
For instance ${\varrho}(L_s-\lambda)= L_s-\chi_1(\tau)$ and ${\varrho}(\partial-\phi_s)=\partial- \tilde{\phi}_s$.
Furthermore, since the isomorphism respects the ring structure, we have
\[L_s-\chi_1(\tau)=(-\partial-\tilde{\phi}_s)(\partial-\tilde{\phi}_s)\]
where $\tilde{\phi}_s$ is a solution of the Ricatti equation $\partial(\phi)+\phi^2=u_s-\chi_s(\tau)$, since $\rho$ respects the differential field structure.

\subsection{Factors at smooth points of  $\Gamma_s$ }\label{sec-specialization}
So far in this paper $\lambda$ and $\mu$ were algebraic variables over $K$, furthermore $\partial \lambda=0$ and $\partial \mu=0$. In this section we will talk about the specialization process of $(\lambda,\mu)$ to a point $P_0=(\lambda_0,\mu_0)$ of the spectral curve $\Gamma_s$. In this manner we recover the classical factorization problem of $L_s-\lambda_0$ as an operator in $K[\partial]$, see for instance \cite{GH}, \cite{AMW}.

\begin{prop}\label{prop-spefact}
Given $P_0=(\lambda_0,\mu_0)$ in $\Gamma_s$ the differential operators $L_s-\lambda_0$ and $A_{2s+1}-\mu_0$ have a common factor over $K$. Furthermore
\begin{equation}\label{eq-Lphi0}
L_s-\lambda_0=(-\partial-\phi_0)(\partial-\phi_0)
\end{equation}
where $\phi_0=\phi_s(P_0)$ with $\phi_s$ as in \eqref{eq-phis} and
\begin{equation}\label{eq-phi0}
\phi_0=-\frac{\varphi_1(P_0)}{\varphi_2(P_0)}=\frac{\mu_0+\alpha(\lambda_0)}{\varphi_2(\lambda_0)}
\end{equation}
with $\varphi_2(\lambda_0)\neq 0$.
\end{prop}
\begin{proof}
By Proposition \ref{prop-Propertiesdres}, 2, the differential operators $L_s-\lambda_0$ and $A_{2s+1}-\mu_0$ in $K[\partial]$ have a common factor since
\[\dres(L_s-\lambda_0,A_{2s+1}-\mu_0)=f_s (\lambda_0,\mu_0)=0.\]
With the notation of Lemma \ref{lem-phi}, observe that $\varphi_1(P_0)+\varphi_2(P_0) \partial$ is the subresultant of $L_s-\lambda_0$ and $A_{2s+1}-\mu_0$ as in Section \ref{subsec-subres}. We will prove next that $\cL_1$ is an operator of order one and then by Theorem \ref{thm-subres}, we have the factorization
\[L_s-\lambda_0=(-\partial-\phi_0)(\partial-\phi_0),\]
where $\phi_0=\phi_s(P_0)$ and the given formula follows by Lemma \ref{lem-phi}.

Let us suppose that $\cL_1$ is the zero operator. Then the second subresultant $\cL_2$ equals to $L_s-\lambda_0$. Hence $L_s-\lambda_0$ is a factor of $A_{2s+1}-\mu_0$. That is
\[A_{2s+1}-\mu_0=Q(L_s-\lambda_0)\]
for some monic differential operator $Q$ of order $2s-1$ in $K[\partial]$.
Computing the commutator with $L_s$ we obtain
\[0=[A_{2s+1}-\mu_0,L_s]=[QL_s,L_s]-[\lambda_0,L_s]=[Q,L_s]L_s.\]
Since $K[\partial]$ is a domain $[Q,L_s]=0$ and $Q$ belong to the centralizer of $L_s$ in $K[\partial]$, which contradicts Theorem \ref{thm-centralizer} since $Q$ has even order less than $2s+1$. We have proved that  $\cL_1$ is an operator of order one, in other words $\varphi_2(\lambda_0)\neq 0$.
\end{proof}

\begin{rem}
Observe that if $\phi_0=0$ then, due to the Ricatti equation, $u_s$ is the constant potential $\lambda_0$, and conversely. From now on we will assume that $u_s$ is not a constant potential.
\end{rem}

We must distinguish two different types of points in the curve, the ones with $\mu_0\neq 0$ and those with $\mu_0= 0$, that is
the finite set
\[Z_s=\Gamma_s\cap (C\times\{0\})=\{(\lambda,0)\mid R_{2s+1}(\lambda)=0\}.\]
Observe that $Z_s$ contains all the affine singular points of $\Gamma_s$.

For a given point $P_0=(\lambda_0,\mu_0)\in C^2$ of the curve $\Gamma_s$, we will assume  $\mu_0\neq 0$ from now on.
Let us consider $\phi_0$ as in \eqref{eq-phi0}, in this section we will use the following notation
\begin{equation}\label{eq-phi+0}
\phi_{0+}=\phi_0=\frac{\mu_0+\alpha(\lambda_0)}{\varphi_2(\lambda_0)}\mbox{ and }\phi_{0-}=\frac{-\mu_0+\alpha(\lambda_0)}{\varphi_2(\lambda_0)},
\end{equation}
pointing out that $\phi_{0+}\neq \phi_{0-}$ since $\mu_0\neq 0$.
Applying Proposition \ref{prop-spefact} to the point  $(\lambda_0,-\mu_0)$ we obtain the following factorization of $L_s-\lambda_0$
\[L_s-\lambda_0=(-\partial-\phi_{0-})(\partial-\phi_{0-}).\]

Let us consider nonzero solutions $\Psi_{0+}$ and $\Psi_{0-}$ respectively of the differential equations
\begin{equation}\label{eq-Psi}
\partial(\Psi)=\phi_{0+}\Psi\mbox{ and }\partial(\Psi)=\phi_{0-}\Psi.
\end{equation}
Then the equality
\[\frac{w(\Psi_{0+},\Psi_{0-})}{\Psi_{0+}\Psi_{0-}}=\phi_{0+}-\phi_{0-}=
\frac{2}{\varphi_2(\lambda_0)}\mu_0\neq 0.\]
implies that $W_0=w(\Psi_{0+},\Psi_{0-})\neq 0$ in $C$. Therefore $\{\Psi_{0+},\Psi_{0-}\}$ is a fundamental set of solutions of $(L_s-\lambda_0)(\Psi)=0$. Moreover
\[\Psi_{0+}\Psi_{0-}=\frac{\varphi_2(\lambda_0)W_0}{2\mu_0}\in K,\]
hence
\begin{equation}\label{eq-PV0}
K\langle \Psi_{0+},\Psi_{0-}\rangle=K\langle \Psi_{0+}\rangle.
\end{equation}

In the next section we will show by means of examples the type of factors that may appear depending on the type of curve. Even at each smooth point of the spectral curve the field $K\langle \Psi_{0+}\rangle $ can be very complicated. These situations deserve a more detailed study that we will present in a future work.

\section{Schr\"odinger operators for KdV solitons. Computed examples }\label{sec-Ejemplos}

Our algorithms \ref{alg-constants} ,for computation of constants, and \ref{alg-fact}, for the factorization of the Schr\"odinger operator, are now ready to be implemented with any symbolic computation software, we did it in Maple 18. We will illustrate their performance by means of three well known families of potentials in \cite{Veselov}. The first one is a family of rational potentials, the second one is a family of Rosen-Morse potentials and both are degenerate cases of a third family of hyperelliptic potentials.

\begin{equation*}
\begin{array}{c|c|c}\label{tab:Solitons}
       \text{Rational}  &   \text{Rosen-Morse} &  \text{Elliptic}\\
       & & \\
        u_s = \dfrac{s(s+1)}{x^2}  & u_s =  \dfrac{-s(s+1)}{\cosh^2(x)} & u_s = s(s+1)\wp(x;g_2,g_3)
\end{array}
\end{equation*}
We will factor $L_s-\lambda$, with $L_s=-\partial^2+u_s$, as an operator in $K(\Gamma_s)[\partial]$, where $\Gamma_s$ is the spectral curve of $L_s$.

\subsection{Rational KdV solitons}

Let us consider the family of rational potentials ${u_s}=s(s+1)/x^2$, $s\geq 1$, in $K=\bbC(x)$ with $\partial=d/dx$. It is well known that the KdV level of $u_s$ is $s$ and its basic constants vector $\bar{c}^s=(0,\ldots ,0)$, we checked this result using Algorithm \ref{alg-constants}.

The spectral curve $\Gamma_s$ is defined by the polynomial $f_s=\mu^2+\lambda^{2s+1}$. We computed the factor $\partial-\phi_s$ of $L_s-\lambda$ in $K(\Gamma_s)[\partial]$ using Algorithm \ref{alg-fact}. For $s=1,2,3$ the results obtained coincides with the ones in \cite{GH}, Example 1.30. We show our result for the next level $s=4$:
\[
\phi_4=\displaystyle{-{\frac {-\mu\,{x}^{9}+10\,{\lambda}^{3}{x}^{6}+270\,{\lambda}^{2}{x}^{4}+4725\,\lambda\,{x}^{2}+44100}{x \left( {\lambda}^{4}{x}^{8}+10\,{\lambda}^{3}{x}^{6}+135\,{\lambda}^{2}{x}^{4}+1575\,\lambda\,{x}^{2}+11025 \right)}}}.
\]
Then, we obtain the factorization:
\[L_4-\lambda=(-\partial-{\phi}_4)(\partial-{\phi}_4)\]
in $K(\Gamma_4 )[\partial]$ where $K(\Gamma_4 )$ is the fraction field of  the domain $K[\lambda,\mu]/(\mu^2+\lambda^{9})$.

Next we observe that the curves $\Gamma_s$ have all genus zero and a global parametrization is $$\aleph_s(\tau)=(\chi_1(\tau),\chi_2(\tau))= \left(-\tau^2,-\tau^{2s+1}\right).$$ Following Section \ref{sec-param},
the one-parameter form of the factor
$ \partial-\tilde{\phi}_4 $
of $$L_4-\chi_1(\tau)=-\partial^2 +\dfrac{20}{x^2 } +\tau^2$$
is given by
\[
\tilde{\phi}_4 (x,\tau )=\displaystyle{-{\frac {{\tau}^{9}{x}^{9}-10\,{\tau}^{6}{x}^{6}+270\,{\tau}^{4}{x}^{4}
\mbox{}-4725\,{\tau}^{2}{x}^{2}+44100}{x \left( {\tau}^{8}{x}^{8}-10\,{\tau}^{6}{x}^{6}+135\,{\tau}^{4}{x}^{4}-1575\,{\tau}^{2}{x}^{2}+11025 \right) }}}.
\]
Then, we obtain the global factorization:
\[L_4-\chi_1(\tau)=(-\partial-\tilde{\phi}_4)(\partial-\tilde{\phi}_4)\]
in $K(\tau)[\partial]=\bbC (x ,\tau )[\partial]$.
A factorization using a global parametrization of the spectral curves is our main contribution to the study of this family of potentials.

\subsection{Rosen-Morse KdV solitons}\label{sec-RM}

Let us consider the family of Rosen-Morse potentials  $u_s=\frac{-s(s+1)}{\cosh^2(x)}$, $s\geq 1$, in the differential field $K=\bbC(e^x)=\bbC\langle\cosh(x)\rangle$ with $\partial=d/dx$.

\para

We show how to obtain the basic constants vector $\bar{c}^s$ for level $s=3$ using Algorithm \ref{alg-constants}. We observe that $u_s$ belongs to $C(\eta)$ with  $\eta=\cosh(x)$ and that $(\eta')^2=\eta^2-1$, thus the hypothesis of the algorithm hold.
For the first three iterations of the algorithm, the system $\cS_n$, $n=0,1,2$ has no solution. In fact, we have
\begin{align*}
&\KdV_0(u_3)=-24\frac{\eta'}{\eta^3},\,\,\,\KdV_1(u_3,c^1)=-24 \frac{\eta'}{\eta^5} \left(-15+ (-1+c_1)\eta^2\right),\\
&\KdV_2(u_3,c^2)=12 \frac{\eta'}{\eta^7} \left(-225+ (30 c_1-150)\eta^2+(-2c_2+2c_1-2)\eta^4\right).
\end{align*}
Thus $\KdV_0(u_3)\neq 0$ and $\KdV_n(u_3,\bar{c}^n)\neq 0$ for all $\bar{c}^n\in\bbC^n$, $n=1,2$.
For $n=3$ we obtain
\begin{align*}
\KdV_3(u_3,c^3)=&\eta'\frac{p_3(\eta)}{q_3(\eta)}\\
=&-12 \frac{\eta'}{\eta^7} \left(-3150+225c_1+ (150 c_1-630-30c_2)\eta^2+(-2c_2+2c_1+2c_3-2)\eta^4\right).
\end{align*}
From the coefficients in $\eta$ of $p_3(\eta)$ we obtain the triangular system
\[\cS_3=\{-3150+225c_1=0,\,\,\,  150 c_1-630-30c_2=0,\,\,\, -2c_2+2c_1+2c_3-2=0\}.\]
The unique solution of this system is the basic constant vector $\bar{c}^3=(14, 49, 36)$. Then, $u_3$ is a solution of the differential equation
\begin{equation*}
\KdV_3(u,\bar{c}^3)=\kdv_3+14 \kdv_{2}+ 49 \kdv_1 +36 \kdv_0 =0.
\end{equation*}

\para

The defining polynomial $f_s$ of $\Gamma_s$ is known to be equal to $f_s=\mu^2+\lambda^2\prod_{\kappa=1}^s(\lambda+\kappa^2)^2$, see for instance \cite{GH}, Example 1.31. We checked these results using our implementation of the differential resultant $\dres(L_s-\lambda, A_{2s+1}-\mu)$.

\para

The next table shows the level $s$, the basic constant vector $\bar{c}^s$, computed with Algorithm \ref{alg-constants}, and the computation of the factor $\partial-\phi_s$ using the Factorization Algorithm \ref{alg-fact} for the operator $L_s -\lambda$ in $K(\Gamma_s)[\partial]$:
\[
\begin{array}{ccc}
s& \bar{c}^s & \phi_s\\
1& (1)  & \displaystyle{\frac{\mu\cosh(x)^3+\sinh(x)}{\cosh(x)(\lambda\cosh(x)^2+\cosh(x)^2-1)}}\\
2& (5,4) &
\displaystyle{\frac{\mu\cosh(x)^5+3\cosh(x)^2\sinh(x)\lambda+12\sinh(x)\cosh(x)^2-18\sinh(x)}
{(\cosh(x)^4\lambda^2+5\cosh(x)^4\lambda+4\cosh(x)^4-3\lambda\cosh(x)^2-12\cosh(x)^2+9)\cosh(x)}}\\
3& (14,49,36) &
\displaystyle{
\frac{\mu+\alpha(\lambda)}{\varphi(\lambda)}}
\end{array}
\]
with
\begin{align*}
&\alpha=\frac{6\cosh(x)^4\sinh(x)\lambda^2+78\cosh(x)^4\sinh(x)\lambda-90\cosh(x)^2\sinh(x)\lambda+a}{\cosh(x)^7},\\
&a=27\sinh(x)(8\cosh(x)^4-30\cosh(x)^2+25),\\
&\varphi=\frac{\cosh(x)^6\lambda^3+14\cosh(x)^6\lambda^2+49\cosh(x)^6\lambda-6\cosh(x)^4\lambda^2-78\cosh(x)^4\lambda+45\cosh(x)^2\lambda+b}{\cosh(x)^6},\\
&b=9\sinh(x)^2(4\cosh(x)^4-20\cosh(x)^2+25).
\end{align*}

\para

All the curves $\Gamma_s$ for this family are rational, in particular they admit a polynomial global parametrization

\begin{equation}\label{eq-param-RM}
\aleph_s(\tau)=(\chi_1(\tau),\chi_2(\tau))=\left(-\tau^2 ,-\tau\prod_{\kappa=1}^s(\tau^2-\kappa^2)
\right).
\end{equation}
The next table shows $\tilde{\phi}_s$:

\[
\begin{array}{cc}
s&  \tilde{\phi}_s\\
1&   \displaystyle{{\frac { \left( {\tau}^{2}-\tau \right) {w}^{2}+ \left( 2\,{\tau}^{2}-4 \right) w+{\tau}^{2}+\tau}{ \left(  \left( \tau-1 \right) w+\tau+1 \right)  \left( w+1 \right)}}}\\
2&  \displaystyle{{\frac { a_3(\tau) {w}^{3}+ b_2(\tau) {w}^{2}+ a_1(\tau) w+a_0(\tau)}{ \left(b_2(\tau) {w}^{2}+ b_1(\tau) w+b_0(\tau)\right) \left( w+1 \right)}}}\\
3&
\displaystyle{{\frac { c_4(\tau) {w}^{4}+ c_3(\tau) {w}^{3}+ c_2(\tau) {w}^{2}+
 c_1(\tau) w+c_0(\tau)}{ \left(  d_3(\tau) {w}^{3}+ d_2(\tau) {w}^{2}+ d_1(\tau) w+d_0(\tau)\right)  \left( w+1 \right) }} }
\end{array}
\]
where $w=e^{2x}$,
\begin{align*}
&a_3=-{\tau}^{3}-3\,{\tau}^{2}-2\,\tau,\,\,\,
a_2=-3\,{\tau}^{3}-3\,{\tau}^{2}+18\,\tau+24,\,\,\,
a_1=-3\,{\tau}^{3}+3\,{\tau}^{2}+18\,\tau-24,\\
&a_0=-{\tau}^{3}+3\,{\tau}^{2}-2\,\tau,\,\,\,
b_2={\tau}^{2}+3\,\tau+2,\,\,\,
b_1= 2\,{\tau}^{2}-8 ,\,\,\,
b_0={\tau}^{2}-3\,\tau+2.
\end{align*}
and
\begin{align*}
&c_4= {\tau}^{4}-6\,{\tau}^{3}+11\,{\tau}^{2}-6\,\tau,\,\,\,
c_3= 4\,{\tau}^{4}-12\,{\tau}^{3}-40\,{\tau}^{2}+168\,\tau-144,\,\,\,
c_2=6\,{\tau}^{4}-102\,{\tau}^{2}+432,\\
&c_1=4\,{\tau}^{4}+12\,{\tau}^{3}-40\,{\tau}^{2}-168\,\tau-144,\,\,\,
c_0={\tau}^{4}+6\,{\tau}^{3}+11\,{\tau}^{2}+6\,\tau,\,\,\,
d_3={\tau}^{3}-6\,{\tau}^{2}+11\,\tau-6,\\
&d_2=3\,{\tau}^{3}-6\,{\tau}^{2}-27\,\tau+54 ,\,\,\,
d_1=3\,{\tau}^{3}+6\,{\tau}^{2}-27\,\tau-54 ,\,\,\,
d_0={\tau}^{3}+6\,{\tau}^{2}+11\,\tau+6 .
\end{align*}
Hence, we obtain the global factorization by means of the global parametrization \eqref{eq-param-RM}:
\[L_s-\chi_1(\tau)=(-\partial-\tilde{\phi}_s)(\partial-\tilde{\phi}_s)\]
in $K(\tau)[\partial]=\bbC (e^x ,\tau )[\partial]$. These factorizations, using a global parametrization of the spectral curves for this family of potentials, are new as far as we know.

\subsection{Elliptic and Hyperelliptic KdV solitons}\label{sec-Elliptic}

Next we consider the family of elliptic potentials $u_s=s(s+1)\wp(x;g_2,g_3)$, $s\geq 1$, where $\wp$ is the Weierstrass $\wp$-function for $g_2$, $g_3$, satisfying $(\wp')^2=4\wp^3-g_2 \wp-g_3$. In this case $K=\bbC\langle\wp\rangle=\bbC(\wp,\wp')$ with $\partial=d/dx$.

\para

The requirements of Algorithm \ref{alg-constants} are satisfied since $u_s$ belongs to $\bbC(\eta)$ for $\eta=\wp$ and $(\eta')^2=4\eta^3-g_2 \eta-g_3\in C(\eta)$. Thus we used Algorithm \ref{alg-constants} to compute $\bar{c}^s$. For $s=1,2$ we could check that the results obtained coincide with the ones in \cite{GH}, Example 1.32.

Next, we show our computations for $s=3$.
Using Algorithm \ref{alg-constants}, we checked that $\KdV_0(u_3)\neq 0$ and $\KdV_n(u_3,\bar{c}^n)\neq 0$ for all $\bar{c}^n\in\bbC^n$, $n=1,2$. From
\begin{align*}
\KdV_3(u_3,c^3)=&\eta'\frac{p_3(\eta)}{q_3(\eta)}\\
=&\eta' \left((-5670g_2-360c_2)\eta^2-2700c_1\eta -153g_2c_1-1782g_3-24c_3\right).
\end{align*}
we obtain the triangular linear system in $c_1$, $c_2$ and $c_3$
\[\cS_3=\{-5670g_2-360c_2=0,\,\,\, 2700c_1=0,\,\,\, -153g_2c_1-1782g_3-24c_3=0\},\]
whose unique solution is $\bar{c}^3=(0, -63g_2/4, -297g_3/4)$. Then, $u_3$ is a solution of the differential equation

\begin{equation*}
\KdV_3(u,\bar{c}^3)=\kdv_3 -\dfrac{63g_2}{4 }\kdv_1 -\dfrac{297g_3 }{4 }\kdv_0 =0.
\end{equation*}

\para

Then, we compute the defining polynomial $f_s$ of $\Gamma_s$ with our implementation of the differential resultant $\dres(L_s-\lambda, A_{2s+1}-\mu)$. Here we obtain the polynomial  $f_3(\lambda,\mu)=\mu^2+R_7(\lambda)$ where
$$R_7=\frac{1}{16}\lambda(-16\lambda^6+504 g_2\lambda^4+2376 g_3\lambda^3-4185 g_2^2\lambda^2+3375 g_2^3-36450 g_2 g_3\lambda-91125 g_3^2).$$

\para

Using Algorithm \ref{alg-fact},
we computed the factor $\partial-\phi_s$ of the operator $L_s -\lambda$ in $K(\Gamma_s)[\partial]$. For $s=1,2$ the results coincide with the ones obtained in \cite{GH}, Example 1.32. We show here
\[
\phi_3= \displaystyle{\frac{\mu+\wp'\left(\frac{675}{2}\wp^2-\frac{225}{8}g_2+45\wp\lambda+3\lambda^2\right)}{\lambda^3+6\wp\lambda^2+(45\wp^2-15g_2)\lambda
-225\wp'^2}}
\]
where $\wp$ and $\wp'$ denote $\wp(x;g_2,g_3)$ and $\wp'(x,g_2,g_3)$ respectively. Then, we obtain the factorization:
\[L_3-\lambda=(-\partial-{\phi}_3)(\partial-{\phi}_3)\]
in $K(\Gamma_3 )[\partial]$ where $K(\Gamma_3 )$ is the fraction field of  the domain $K[\lambda,\mu]/(\mu^2+R_7 (\lambda))$.

\para

It is well known that the curves $\Gamma_s$ for this family are not rational, they have genus $s$.  In the case of the elliptic potential $u_1=2\wp(x;g2,g3)$ one can easily prove that $\aleph_1(\tau)=\left(-\wp(\tau),\frac{1}{2}\wp'(\tau)\right)$ is a global parametrization of the spectral curve $\Gamma_1$ whose defining polynomial is the irreducible polynomial $f_1=-\mu^2-\lambda^3+(1/4)g_2\lambda -(1/4)g_3$ . In this case
\[ \tilde{\phi}_1=\displaystyle{\frac{\frac{-1}{2}(\wp'(x)-\wp'(\tau))}{\wp(x)-\wp(\tau)}}.\]
Hence, we obtain the global factorization by means of the given global parametrization:
\[L_1+\wp(\tau)=(-\partial-\tilde{\phi}_1)(\partial-\tilde{\phi}_1)\]
in $\bbC \langle \wp (x)  ,\wp (\tau ) \rangle[\partial]$. For $s\geq 2$, as far as we know there are no effective algorithms to compute a global parametrization $(\chi_1(\tau),\chi_2(\tau))$ of $\Gamma_s$. This is a difficult open problem. Some contributions have been made in this direction, for instance by Y.V. Brezhnev in \cite{Brez2}.

\bigskip

\noindent {\sf Acknowledgments:} We kindly thank all members of the Integrability Madrid Seminar for many fruitful discussions: J. Capit\'an, R. Hern\'andez Heredero, S. Jim\'enez, A. P\' erez-Raposo, J. Rojo Montijano and R. S\'anchez; and the members in Colombia: P.B. Acosta-Hum\' anez and D. Bl\' azquez-Sanz. In particular: to P.B. Acosta-Hum\' anez for his hospitality and enlighting discussions during the visit of the first two authors to Universidad del Atl\' antico, Barranquilla, Colombia, 2014; to R. Hern\'andez Heredero for showing us the importance of the recursion operator and the proof of Lemma \ref{lem-kdv}; and to A. P\' erez-Raposo for carefully proof reading this manuscript. We also thank A. Mironov and A.P. Veselov for stimulating discussion on this kind of problems.

\para

The first two authors are members of the Research Group \textquotedblright Modelos matem\' aticos no lineales\textquotedblright , UPM and S.L. Rueda has been
partially supported by the \textquotedblright Ministerio de Econom\'\i a y Competitividad\textquotedblright  under the project MTM2014-54141-P. M.A. Zurro is partially supported by Grupo UCM 910444.

\end{document}